\documentclass[letterpaper,10pt,conference]{ieeeconf}

\IEEEoverridecommandlockouts
\usepackage[T1]{fontenc}
\usepackage{amsmath,amssymb,mathtools,bm}

\usepackage{amsthm}
\usepackage{graphicx}
\graphicspath{{figures/}}
\usepackage[caption=false,font=footnotesize]{subfig}
\usepackage{booktabs}
\usepackage{microtype}
\usepackage{balance}
\usepackage{comment}
\usepackage{hyperref}
\hypersetup{hidelinks}

\newtheorem{assumption}{Assumption}
\newtheorem{lemma}{Lemma}
\newtheorem{definition}{Definition}
\newtheorem{theorem}{Theorem}
\newtheorem{corollary}{Corollary}
\newtheorem{proposition}{Proposition}

\DeclareMathOperator*{\argmin}{arg\,min}
\newcommand{\R}{\mathbb R}
\newcommand{\C}{\mathcal C}
\newcommand{\U}{\mathcal U}
\newcommand{\Hh}{\mathcal H}

\title{History-Aware Control Barrier Functionals with Nonsingular Fading-Memory Kernels}

\author{Ege~C.~Altunkaya, Esra Demir, Murad Abu-Khalaf, and \.{I}brahim~\"{O}zkol%
\thanks{E. C. Altunkaya, E. Demir, and \.{I}. \"{O}zkol are with the Aviation Institute, Istanbul Technical University, 34469 Istanbul, T\"{u}rkiye. M. Abu-Khalaf is with Turkish Aerospace, Istanbul, T\"{u}rkiye. Corresponding author: Ege C. Altunkaya (e-mail: altunkaya16@itu.edu.tr).}}

\begin{document}

\maketitle

\begin{abstract}
Safety-critical protection often depends not only on the state and physical limits of the system, but on how that state evolved over time. Here, the safety margin measures the remaining allowance before a physical limit is reached. In this context, the objective is to enforce both physical safety and history-dependent retention of this margin for integer-order control-affine systems. To this end, changes in this margin are converted into memory through fading-memory kernels with finite nonnegative weights. A history-aware control barrier functional (HA-CBFal) is constructed explicitly, with nonnegativity equivalent to retaining a prescribed fraction of this margin's weighted historical value. This condition is enforced by an affine input constraint and proved to imply physical safety under stated kernel, initialization, and feasibility conditions, without finite-dimensional memory realization. In 36 calibration cases on the ADMIRE flight-control benchmark, mean per-case tracking-RMSE reductions of 27.2\% and 7.0\% are obtained relative to calibrated state-only CBF and moving-average control barrier functional (MA-CBFal), respectively, with both constraints preserved to numerical tolerance. In a shorter-recovery maneuver outside calibration with frozen gains, both constraints are preserved to tolerance by HA-CBFal and state-only CBF, but the retention requirement is violated despite MA-CBFal nonnegativity.
\end{abstract}

\begin{keywords}
Control barrier functions, control barrier functionals, fading memory, convolutional memory, history-dependent safety.
\end{keywords}

\setlength{\abovedisplayskip}{6pt plus 2pt minus 2pt}
\setlength{\belowdisplayskip}{6pt plus 2pt minus 2pt}

\section{Introduction}
\label{sec:introduction}

Physical operating limits must be respected during task execution. For some
protection tasks, the preceding operation must also be considered. For
example, motor restart permission is conditioned on thermal loading and
cooling \cite{abb2015umc}, and strain-history effects are retained in
viscoelastic stresses \cite{jang2023viscoelastic}. In flight-envelope
protection at high incidence, adequate nose-down recovery capability is required,
while aerodynamic response can be influenced by recent pitching motion
\cite{lackey1992pitchmargin,brandon1991dynamicstall}. Thus, even with a fixed
physical limit, additional protection may be prescribed according to the
preceding evolution of physical quantities, such as temperature, strain, or
incidence, termed \emph{history}.

In safety-critical control, the satisfaction of physical limits is formalized
as safe-set forward invariance. To this end, the requirement is translated
through control barrier functions (CBFs) into input constraints enforced by a
safety filter around a nominal controller \cite{ames2017cbfqp,ames2019cbf}.
A state-only CBF determines the certified input set from the plant state. Even when the current plant state and control input are sufficient to determine the next state, enforcing safety can further benefit from leveraging how the plant state evolved over time.

To isolate the role of history, the safety margin $h$ measures the remaining
allowance to the physical limit: $h=0$ marks the limit and $h<0$ its violation.
The admissible-input set and preceding safety-margin history are denoted by
$\U$ and $\Hh_T$, respectively.
Two admissible trajectories are considered, with the same checkpoint
$x^\star$ and safety margin reached at time $T$,
\begin{equation}
x_{\rm A}(T)=x_{\rm B}(T)=x^\star,\qquad
h_{\rm A}(T)=h_{\rm B}(T)=h^\star .
\label{eq:same_terminal_barrier}
\end{equation}
Their preceding safety-margin evolutions need not be identical. In
Fig.~\ref{fig:history_geometry}(a), the recent
evolution associated with $A$ is taken to be more adverse. A protection
policy may be specified by assigning more restrictive permitted actions
after adverse evolution, as illustrated by the ordering
\begin{equation}
\U_{\rm H}(x^\star,\Hh_T^{\rm A})
\subsetneq
\U_{\rm H}(x^\star,\Hh_T^{\rm B})
\subseteq\U .
\label{eq:generic_authority_ordering}
\end{equation}
A policy is termed \emph{history-aware} when different permitted actions
are assigned to different histories reaching the same plant state.
Both sets cannot be reproduced at the same $x^\star$ by a history-blind
mapping $\U_{\rm s}(x)$. A common history-blind set must be contained in
their intersection and is therefore more restrictive for history $B$.
A control constraint must therefore be constructed to enforce the prescribed
history dependence while preserving the physical safety limit.

\begin{figure}[!t]
 \centering
 \includegraphics[width=\columnwidth]{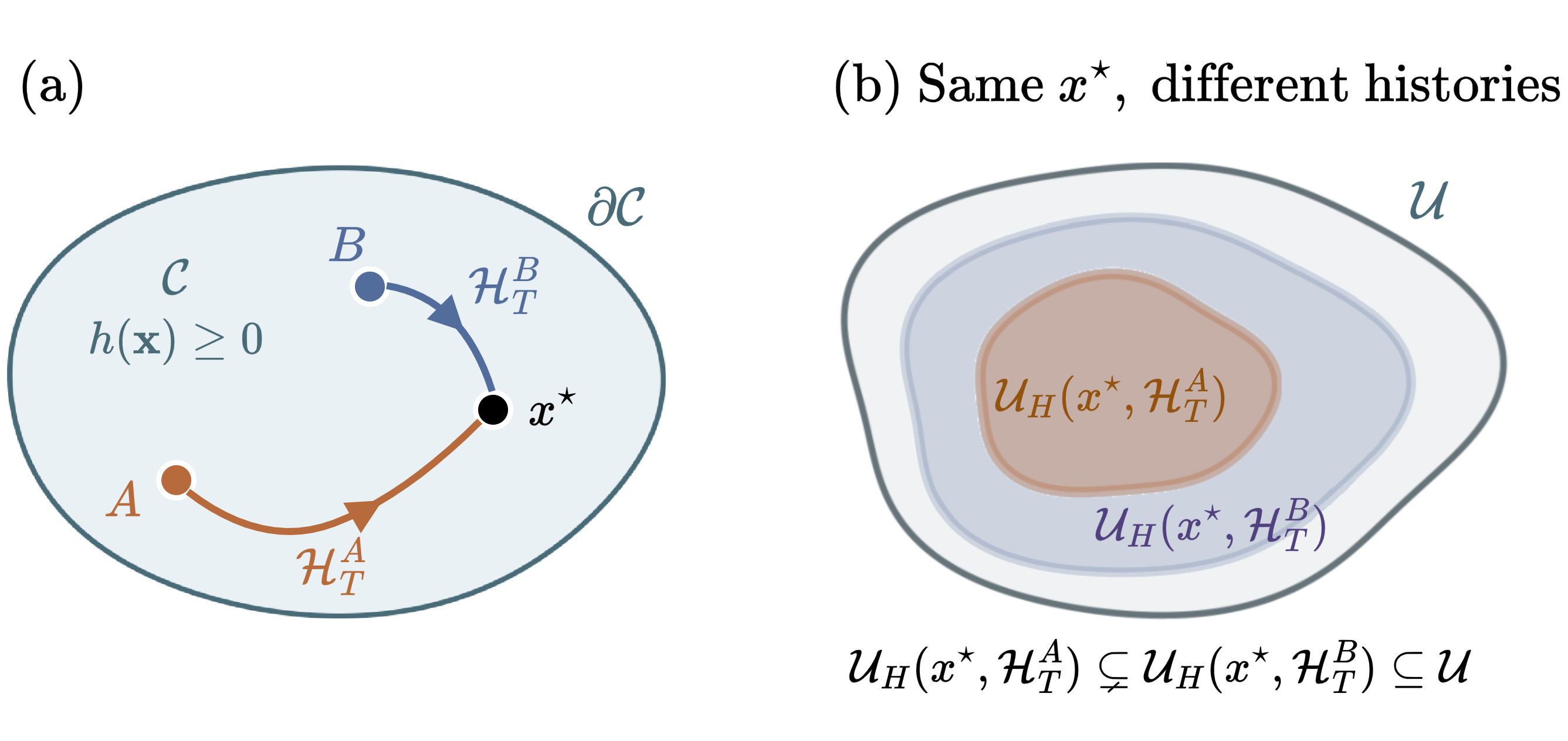}
 \caption{History-conditioned control authority. (a) Distinct histories reach
 the same state $x^\star$ and margin $h^\star$, with different retained
 memories. Path geometry is illustrative; path length is not the memory
 criterion. (b) Different admissible-control sets may be assigned at that
 state by a history-aware protection policy.}
 \label{fig:history_geometry}
\end{figure}

Related uses of history are found in delay-system safety analysis. The state
segment over the delay interval is needed to determine subsequent evolution,
and safety certificates have been developed through Razumikhin and Krasovskii
approaches \cite{prajna2005delay,ren2021razumikhin,ren2022safe}. Similarly,
past evolution is embedded in fractional-order plant dynamics through the
fractional derivative, for which barrier conditions have been developed
\cite{zhu2023first,zhu2025caputo}. In these settings, historical dependence
in the plant dynamics must be accommodated in the safety analysis. In the
example above, history-dependent restrictions are instead prescribed by the
protection policy.

History is also admitted into the safety specification itself through control
barrier functionals (CBFals): sets defined by current, delayed, or
distributed-state quantities can be rendered invariant
\cite{orosz2019safety,kiss2023cbf}. Within this framework, the functional
must be constructed to encode the prescribed protection policy. Physical
safety must then be established for that construction: a bound imposed on a
historical average need not be respected by the current value.

In this study, physical safety and a prescribed history-dependent reserve
are considered jointly for an integer-order, delay-free plant. The physical
limit is expressed through a safety margin. To distinguish recent depletion
from recovery, past decreases and increases of this margin are accumulated
into a quantity termed \emph{memory}. This accumulation is performed through
a nonsingular fading-memory kernel: each margin change is assigned a
nonnegative weight, finite at zero age and nonincreasing toward zero with
age. The weighted historical value of the same safety margin is then obtained
from this memory. A \emph{reserve policy} requires the current safety margin
to retain a prescribed fraction of that value. Both nonnegativity of the safety margin
and satisfaction of this reserve requirement are imposed as mandatory
constraints.

Against this background, the following research question is posed:
how can a prescribed fading-memory reserve requirement be systematically converted
into a control constraint that also certifies physical safety?
To this end, a history-aware CBFal (HA-CBFal) and its affine input constraint are
constructed directly from the prescribed safety margin, fading kernel, and
reserve fraction. Auxiliary-barrier nonnegativity is made equivalent to
reserve satisfaction and enforced through standard CBFal conditions. Under
stated assumptions, physical safety is established across the admissible
kernel family without a separate implication proof for each kernel choice.
The contributions are:
\begin{itemize}
\item An explicit certificate construction is provided for prescribed
fading-memory reserve requirements. Joint reserve and physical-safety preservation is
proved through a positive Volterra resolvent under explicit kernel,
initialization, and feasible-enforcement conditions.
\item Certified input sets at a fixed state are characterized through a
scalar history shift. Their ordering and exact worst-history intersection
are established, identifying the greatest history-blind set satisfying the
same certificate over a prescribed history family.
\item In ADMIRE \cite{forssell2005admire}, tracking benefits and the distinction between functional
preservation and reserve satisfaction are demonstrated against CBF/CBFal
comparators calibrated under both requirements, using Gaussian memory with
no exact finite-dimensional linear time-invariant (LTI) realization.
\end{itemize}

\section{Preliminaries}
\label{sec:preliminaries}

\subsection{Control Barrier Functions}
A state domain $\mathcal D\subset\R^n$ and a closed convex
admissible-input set $\U\subset\R^m$ are considered for the
control-affine system
\begin{equation}
\dot x=f(x)+g(x)u,\qquad x\in\mathcal D,\quad u\in\U,
\label{eq:plant}
\end{equation}
where $f$ and $g$ are locally Lipschitz. Controller memory is excluded from the Markovian plant state $x$. The physical safe set is defined by $h\in C^1(\mathcal D,\R)$ as
\begin{equation}
\C\triangleq\{x\in\mathcal D:h(x)\geq0\},
\label{eq:physical_safe_set}
\end{equation}
with $\dot h=L_fh+L_gh\,u$ along \eqref{eq:plant}. For the
relative-degree-one setting considered, the standard linear CBF condition
is
\begin{equation}
\sup_{u\in\U}
\left[L_fh(x)+L_gh(x)u+kh(x)\right]\geq0,
\qquad k>0.
\label{eq:classical_cbf_definition}
\end{equation}
Under any admissible feedback satisfying the corresponding pointwise
inequality, $h(t)\geq e^{-kt}h(0)$ is obtained, and $\C$ is rendered forward
invariant under the usual closed-loop existence assumptions \cite{ames2017cbfqp}.
A certified-input map $\U_{\rm s}(x)$ is history-blind if it depends only
on $x$, independently of retained margin history.

\subsection{Fading-Memory Kernels}
For age $s=t-\tau\geq0$, the contribution generated at time $\tau$
is weighted by a causal kernel $\kappa(s)$. A continuous nonsingular
fading weight is represented by $\kappa\in C^0(\R_{\geq0};\R_{\geq0})$
with $\kappa(0)=\kappa_0\in(0,\infty)$ and
$\lim_{s\to\infty}\kappa(s)=0$.
Choose the kernel's shape and decay time for the desired protection horizon,
subject to Assumption~\ref{ass:kernel} and input feasibility.

\subsection{Convolutional Memory}
For causal scalar signals, convolution is denoted by
$(a*b)(t)\triangleq\int_0^t a(t-\tau)b(\tau)\,d\tau$. For an absolutely
continuous $y:[0,T]\to\R$, the convolutional memory operator generated by
the fading kernel $\kappa$ of Section~II-B is defined as \cite{sabatier2020continuous}
\begin{equation}
\bigl(\mathfrak D_{\kappa}y\bigr)(t)
\triangleq
(\kappa*\dot y)(t)
=
\int_0^t\kappa(t-\tau)\dot y(\tau)\,d\tau .
\label{eq:general_memory_operator}
\end{equation}
Past signal changes are weighted by the kernel. Only the kernel conditions in
Assumption~\ref{ass:kernel} are required; fractional kernels are optional
specializations. For the physical barrier, the memory is denoted by
$q_\kappa(t)\triangleq\bigl(\mathfrak D_\kappa h\bigr)(t)$, with
$q_\kappa(0)=\bigl(\mathfrak D_\kappa h\bigr)(0)=0$ at the memory origin, and
its retained history is denoted by $\Hh_t\triangleq h|_{[0,t]}$.

\section{Main Result}
\label{sec:main_result}

An input-affine constraint with a history-dependent offset certifies
physical safety without finite-dimensional memory realization.

\begin{assumption}[Admissible fading-memory kernel]
\label{ass:kernel}
Local absolute continuity is assumed for
$\kappa:\R_{\geq0}\rightarrow\R_{\geq0}$, with $\kappa(0)=\kappa_0\in(0,\infty)$,
$\dot\kappa(t)\leq0$ for almost every (a.e.) $t\geq0$, and
$\lim_{t\to\infty}\kappa(t)=0$. The weight function is defined by
$\omega(t)\triangleq-\dot\kappa(t)$. Then
$\omega\in L^1_{\rm loc}(\R_{\geq0})$, $\omega\geq0$ a.e., and, for every
$T\geq0$, $\int_0^T\omega(s)\,ds=\kappa_0-\kappa(T)\leq\kappa_0$.
Consequently, $\omega\in L^1(\R_{\geq0})$ and
$\|\omega\|_1=\kappa_0$.
\end{assumption}

The operator \eqref{eq:general_memory_operator} is henceforth evaluated
under Assumption~\ref{ass:kernel}, with the memory horizon measured
from its lower terminal $0$. Here, $h(t)\equiv h(x(t))$ denotes the safety margin. For
$\gamma>0$, the auxiliary history-dependent barrier is defined as
\begin{equation}
\psi_\kappa(t)
\triangleq
\bigl(\mathfrak D_\kappa h\bigr)(t)+\gamma h(t)
=
q_\kappa(t)+\gamma h(t).
\label{eq:general_auxiliary_barrier}
\end{equation}

For $q_\kappa$ defined after \eqref{eq:general_memory_operator}, local absolute continuity gives
$\dot q_\kappa=\kappa_0\dot h+\dot\kappa*\dot h$ a.e.
Thus, along \eqref{eq:plant},
\begin{align}
\dot\psi_\kappa
={}&
(\kappa_0+\gamma)
\bigl(L_fh(x)+L_gh(x)u\bigr)
\notag\\
&+
\int_0^t
\dot\kappa(t-\tau)\dot h(\tau)\,d\tau ,
\quad\text{a.e.}
\label{eq:general_auxiliary_derivative}
\end{align}

\begin{definition}[Generalized-kernel history-aware CBFal]
\label{def:general_hacbf}
Under Assumption~\ref{ass:kernel}, the auxiliary barrier
$\psi_\kappa$ in \eqref{eq:general_auxiliary_barrier} is a generalized-kernel
history-aware control barrier functional for $h$ on the considered state-history
domain if there exists $\beta>0$ such that, at each state-history pair,
\begin{equation}
\exists u\in\U:\quad Q_\kappa(x,\Hh_t,u)\geq0,
\label{eq:general_hacbf_definition}
\end{equation}
where
\begin{align}
Q_\kappa(x,\Hh_t,u)
\triangleq{}&
(\kappa_0+\gamma)
\bigl(L_fh(x)+L_gh(x)u\bigr)
\notag\\
&+
\int_0^t
\dot\kappa(t-\tau)\dot h(\tau)\,d\tau
+\beta\psi_\kappa(t).
\label{eq:general_hacbf_row}
\end{align}
Equivalently,
$Q_\kappa=\dot\psi_\kappa+\beta\psi_\kappa$.
\end{definition}

\begin{lemma}[Positive $L^1$ resolvent]
\label{lem:positive_volterra}
A weight $\omega\in L^1(\R_{\geq0})$ with $\omega\geq0$ a.e. is
considered, and $\lambda\geq0$ is chosen such that
$\rho\triangleq\lambda\|\omega\|_1<1$. For a fixed kernel $\omega$, the Volterra operator is defined by
$(\mathcal V_\omega\varphi)(t)\triangleq(\omega*\varphi)(t)$, with
$\mathcal V_\omega^0\triangleq\mathcal I$. Convolution powers and the
resolvent are defined by $\omega^{(*1)}\triangleq\omega$,
$\omega^{(*(j+1))}\triangleq\omega*\omega^{(*j)}$, and
\begin{equation}
r_\lambda
\triangleq
\sum_{j=1}^{\infty}\lambda^j\omega^{(*j)} .
\label{eq:volterra_resolvent_kernel}
\end{equation}
Then $r_\lambda\in L^1(\R_{\geq0})$, $r_\lambda\geq0$ a.e., and
$\|r_\lambda\|_1\leq\rho/(1-\rho)$. Moreover, for every
$f\in L^\infty_{\rm loc}(\R_{\geq0})$, on every finite interval the
second-kind Volterra equation $y=f+\lambda\mathcal V_\omega y$ has the unique
bounded solution $y=f+r_\lambda*f$. Hence $f\geq0$ implies $y\geq0$.
\end{lemma}

\begin{proof}
    The proof is obtained from the successive-approximation construction and
    resolvent formula in \cite[Sec.~2.3, Thms.~3.1 and~3.5, pp.~42--44]{GripenbergLondenStaffans1990},
    with $k=-\lambda\omega$; the stated positivity and norm bound are obtained
    from the nonnegative convolution series.
\end{proof}

\begin{theorem}[Main result: auxiliary-to-physical invariance]
\label{thm:physical_transfer}
Under Assumption~\ref{ass:kernel}, $\gamma>0$ and $h(x(0))\geq0$ are
assumed. Along a complete absolutely continuous closed-loop solution of
\eqref{eq:plant}, the applied input is assumed to satisfy
\begin{equation}
Q_\kappa(x(t),\Hh_t,u(t))\geq0
\quad\text{a.e.}
\label{eq:theorem_hacbf_condition}
\end{equation}
for some $\beta>0$. Then $\psi_\kappa(t)\geq0$ and $h(x(t))\geq0$, $t\geq0.$
Consequently, the physical safe set $\C$ in
\eqref{eq:physical_safe_set} is forward invariant.
\end{theorem}

\begin{proof}
\emph{Step 1: Auxiliary-barrier nonnegativity.} At the memory origin,
$\bigl(\mathfrak D_\kappa h\bigr)(0)=0$, so
$\psi_\kappa(0)=\gamma h(0)\geq0$. From
\eqref{eq:general_hacbf_row} and \eqref{eq:theorem_hacbf_condition},
$\dot\psi_\kappa+\beta\psi_\kappa\geq0$ a.e.; the integrating factor
$e^{\beta t}$ is therefore applied to obtain
\begin{equation}
\psi_\kappa(t)\geq e^{-\beta t}\psi_\kappa(0)
=\gamma e^{-\beta t}h(0)\geq0.
\label{eq:psi_nonnegative}
\end{equation}

\emph{Step 2: Volterra reformulation.} Since $\kappa$ and $h$ are
absolutely continuous on every finite interval,
\eqref{eq:general_memory_operator} is integrated by parts to obtain
\begin{align}
q_\kappa(t)
&=\left[\kappa(t-\tau)h(\tau)\right]_{\tau=0}^{\tau=t}
 +\int_0^t\dot\kappa(t-\tau)h(\tau)\,d\tau \notag\\
&=\kappa_0h(t)-\kappa(t)h(0)-(\omega*h)(t).
\label{eq:general_integration_by_parts}
\end{align}
Hence
\begin{equation}
(\kappa_0+\gamma)h(t)
=\psi_\kappa(t)+\kappa(t)h(0)+(\mathcal V_\omega h)(t).
\label{eq:physical_identity_with_I}
\end{equation}
Since the sign of $h$ is not yet established,
$\lambda_\gamma\triangleq(\kappa_0+\gamma)^{-1}$ and
$f(t)\triangleq\lambda_\gamma[\psi_\kappa(t)+\kappa(t)h(0)]$ are defined.
By \eqref{eq:psi_nonnegative}, Assumption~\ref{ass:kernel}, and $h(0)\geq0$,
$f\geq0$ is ensured, and \eqref{eq:physical_identity_with_I} is rewritten as
\begin{equation}
h=f+\lambda_\gamma\mathcal V_\omega h.
\label{eq:physical_volterra_equation}
\end{equation}

\emph{Step 3: Positive reconstruction.} Lemma~\ref{lem:positive_volterra} is applied.

The first successive substitutions are written as
\begin{align*}
h&=f+\lambda_\gamma\mathcal V_\omega f
 +\lambda_\gamma^2\mathcal V_\omega^2h\\
 &=f+\lambda_\gamma\mathcal V_\omega f
 +\lambda_\gamma^2\mathcal V_\omega^2f
 +\lambda_\gamma^3\mathcal V_\omega^3h.
\end{align*}
Hence, for every integer $N\geq0$,
\begin{equation}
h=\sum_{j=0}^{N}\lambda_\gamma^j\mathcal V_\omega^{\,j}f
 +\lambda_\gamma^{N+1}\mathcal V_\omega^{\,N+1}h.
\label{eq:finite_N_volterra_expansion}
\end{equation}
The finite sum is the order-$N$ Neumann partial sum; the last term is its
remainder, which still contains the unknown $h$.
Under Assumption~\ref{ass:kernel}, the contraction factor is obtained as
\begin{equation}
\lambda_\gamma\|\omega\|_1
=\frac{\kappa_0}{\kappa_0+\gamma}\triangleq\theta<1.
\label{eq:physical_contraction_bound}
\end{equation}
On every finite $[0,T]$, $h$ is bounded by continuity; by Young's inequality, 
$\bigl\|\lambda_\gamma^{N+1}\mathcal V_\omega^{\,N+1}h\bigr\|_{\infty,[0,T]}
\leq\theta^{N+1}\|h\|_{\infty,[0,T]}\longrightarrow0$. Thus the partial sums converge uniformly to $h$ on every finite interval.
With $\mathcal V_\omega^{\,j}f=\omega^{(*j)}*f$ for $j\geq1$, the limit is written as
\begin{equation}
h=\sum_{j=0}^{\infty}\lambda_\gamma^j\mathcal V_\omega^{\,j}f
 =f+r_{\lambda_\gamma}*f,
\label{eq:neumann_limit_reconstruction}
\end{equation}
where $r_{\lambda_\gamma}\geq0$ by Lemma~\ref{lem:positive_volterra}.
The explicit safety margin is obtained by substituting $f$:
\begin{equation}
\begin{aligned}
h(t)={}&\frac{\psi_\kappa(t)}{\kappa_0+\gamma}
 +\frac{\kappa(t)h(0)}{\kappa_0+\gamma}\\
&+\frac{1}{\kappa_0+\gamma}
 \int_0^t r_{\lambda_\gamma}(t-\tau)\psi_\kappa(\tau)\,d\tau\\
&+\frac{h(0)}{\kappa_0+\gamma}
 \int_0^t r_{\lambda_\gamma}(t-\tau)\kappa(\tau)\,d\tau.
\end{aligned}
\label{eq:explicit_positive_reconstruction}
\end{equation}
Every term is nonnegative because $\psi_\kappa\geq0$, $\kappa\geq0$,
$r_{\lambda_\gamma}\geq0$, and $h(0)\geq0$. Hence $h(x(t))\geq0$ is ensured for all $t\geq0$, and forward invariance
of $\C$ is established.
\end{proof}

\begin{corollary}[Kernel-weighted reserve interpretation]
\label{rem:kernel_policy}
With $\omega=-\dot\kappa$, the weighted historical value of the safety margin is
\begin{equation}
\bar h_\kappa(t)\triangleq
\frac{\kappa(t)h(0)+(\omega*h)(t)}{\kappa_0}.
\label{eq:kernel_historical_margin}
\end{equation}
Nonnegative weights are ensured by Assumption~\ref{ass:kernel}, and
$\kappa(t)+\int_0^t\omega(s)\,ds=\kappa_0$, so $\bar h_\kappa$ is a
normalized weighted historical value of $h$. From \eqref{eq:general_integration_by_parts},
$\mathfrak D_\kappa h=\kappa_0(h-\bar h_\kappa)$ is obtained. With
$\theta\triangleq\kappa_0/(\kappa_0+\gamma)\in(0,1)$ and
the \emph{reserve-policy residual} $g_{\kappa,\theta}\triangleq h-\theta\bar h_\kappa$, the identities are
$\psi_\kappa=(\kappa_0+\gamma)g_{\kappa,\theta}$ and
$Q_\kappa=(\kappa_0+\gamma)(\dot g_{\kappa,\theta}+\beta g_{\kappa,\theta})$.
Thus $g_{\kappa,\theta}\geq0$ expresses the reserve requirement $h\geq\theta\bar h_\kappa$ for each admissible kernel. Conversely, for a prescribed $\theta\in(0,1)$,
$\gamma=\kappa_0(1-\theta)/\theta>0$ is obtained. Moreover, by
\eqref{eq:physical_contraction_bound}, the reserve fraction $\theta$ is also
the Neumann-resolvent contraction factor used in
Theorem~\ref{thm:physical_transfer}. A finite-dimensional realization of the history functional is not
required by the theorem.
Under the hypotheses of Theorem~\ref{thm:physical_transfer},
$g_{\kappa,\theta}\geq0$ and $h\geq0$ therefore hold.
\end{corollary}

With $g_{\kappa,\theta}$ as the functional, standard CBFal enforcement is
recovered. For prescribed $h,\kappa,\theta$, the auxiliary barrier is
obtained through $\gamma=\kappa_0(1-\theta)/\theta$.
For fixed $h$, the affine form and input coefficient $L_gh$ are retained
in $Q_\kappa/(\kappa_0+\gamma)$; policy terms are updated analytically.
Under Theorem~\ref{thm:physical_transfer}, physical safety is certified
across admissible kernels without a separate implication proof for each choice.

\begin{proposition}[History-parametrized certified input sets]
\label{prop:history_parametrized_sets}
The scalar history functional is defined as
\begin{equation}
\Lambda(\Hh_t)\triangleq
\int_0^t[\beta\kappa(t-\tau)+\dot\kappa(t-\tau)]\dot h(\tau)\,d\tau .
\label{eq:history_functional}
\end{equation}
For fixed $x$, $a(x)\triangleq(\kappa_0+\gamma)L_gh(x)$ and
$b(x)\triangleq(\kappa_0+\gamma)L_fh(x)+\beta\gamma h(x)$ are defined,
and $\mathfrak H(x)$ is a prescribed nonempty family of margin histories
induced by admissible plant trajectories ending at $x$. Then
$Q_\kappa=a(x)u+b(x)+\Lambda(\Hh)$ and
$\U_{\rm H}(x,\Hh)=\U\cap\{u:a(x)u+b(x)+\Lambda(\Hh)\geq0\}$.
Thus identical certified input sets are obtained for equal values of $\Lambda$, while
$\Lambda(\Hh^{\rm A})\leq\Lambda(\Hh^{\rm B})$ implies
$\U_{\rm H}(x,\Hh^{\rm A})\subseteq\U_{\rm H}(x,\Hh^{\rm B})$. Moreover, if
$\underline\Lambda(x)\triangleq\inf_{\Hh\in\mathfrak H(x)}\Lambda(\Hh)>-\infty$,
then
\begin{equation}
\begin{aligned}
\U_{\rm wc}(x)&\triangleq\bigcap_{\Hh\in\mathfrak H(x)}\U_{\rm H}(x,\Hh)\\
&=\U\cap\{u:a(x)u+b(x)+\underline\Lambda(x)\geq0\}.
\end{aligned}
\label{eq:uniform_history_intersection}
\end{equation}
Hence $\U_{\rm wc}(x)$ is the greatest history-blind input set satisfying the
same certificate for every history in $\mathfrak H(x)$.
\end{proposition}

\begin{proof}
The expression for $Q_\kappa$ is obtained by collecting its two
history-convolution terms. Equality and monotone inclusion are
established because only the scalar half-space offset is changed by history. Finally,
$u\in\bigcap_{\Hh}\U_{\rm H}(x,\Hh)$ iff $u\in\U$ and
$a(x)u+b(x)+\Lambda(\Hh)\geq0$ for every $\Hh\in\mathfrak H(x)$, which is
equivalent to $a(x)u+b(x)+\underline\Lambda(x)\geq0$.
\end{proof}

Thus $\Lambda$ is sufficient to determine the instantaneous certified input
set at fixed $x$; a closed scalar realization of the memory evolution
is not asserted. If $a(x)\neq0$, the unconstrained half-space boundaries are
parallel, with separation
$|\Lambda(\Hh^{\rm B})-\Lambda(\Hh^{\rm A})|/\|a(x)\|_2$; strict set inclusion is obtained only when the intervening slab
also intersects $\U$. The ordering in Fig.~\ref{fig:history_geometry} is thereby realized while
$\C$ is left unchanged.

\subsection{History-Aware Safety Filter}
\label{subsec:general_hacbf_qp}

For locally essentially bounded closed-loop inputs, the convolution terms
defining $\Lambda(\Hh_t)$ are continuous in $t$ under
Assumption~\ref{ass:kernel}. For a nominal input $u_{\rm nom}$ and $W\succ0$, the generalized-kernel
history-aware CBFal can be enforced through the minimum-intervention program
\begin{equation}
\begin{aligned}
u^*(x,\Hh_t)
\triangleq
\argmin_{u\in\U}\quad&
\frac{1}{2}\|u-u_{\rm nom}\|_W^2\\
\mathrm{s.\ t.}\quad&
Q_\kappa(x,\Hh_t,u)\geq0,
\end{aligned}
\label{eq:general_hacbf_qp}
\end{equation}
where $Q_\kappa$ is given by \eqref{eq:general_hacbf_row}. Because
$Q_\kappa$ is affine in $u$, \eqref{eq:general_hacbf_qp} is a convex
program, and a QP when $\U$ is polyhedral.

\begin{figure*}[!t]
 \centering
 \includegraphics[width=\textwidth]{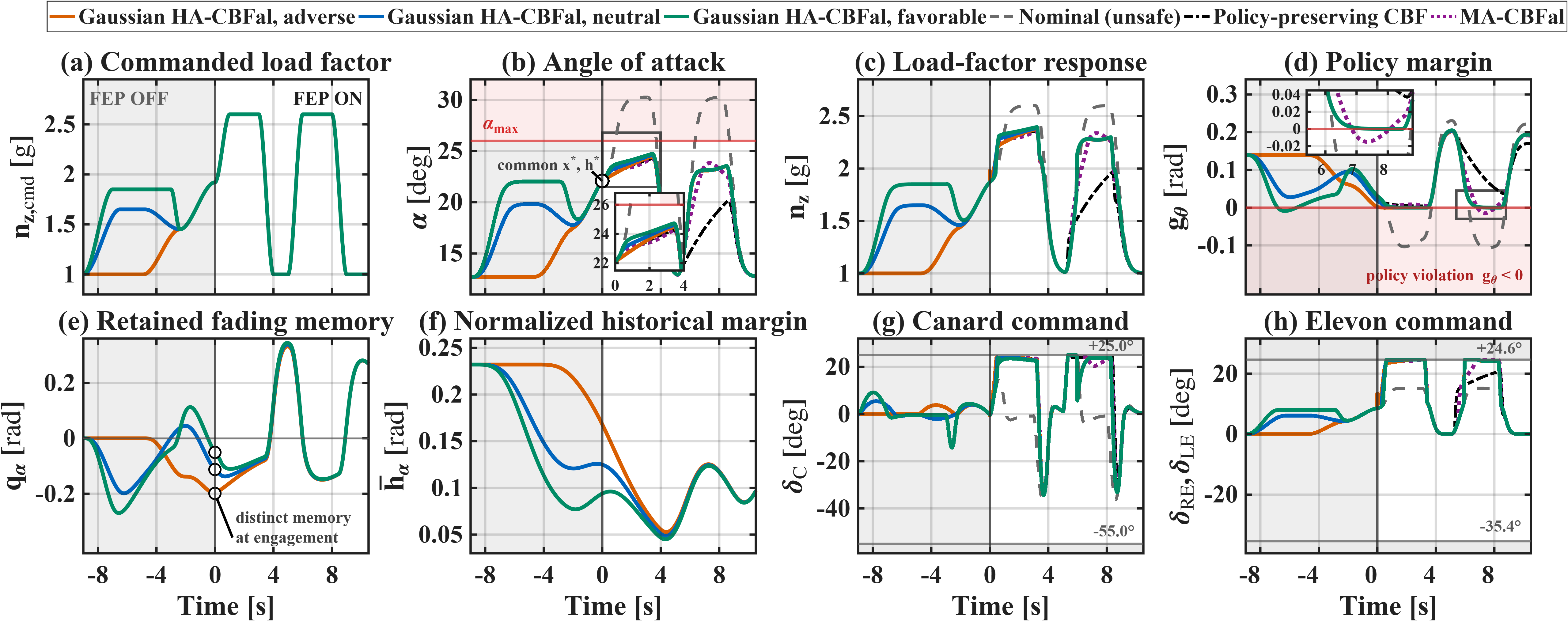}
 \caption{Stitched preparation/activation scenario.
 (a)--(d) Command, angle of attack, load-factor response, and reserve-policy
 residual $g_\theta$ (panel label: policy margin); (e)--(f) retained Gaussian
 memory and weighted historical value $\bar h_\alpha$ (panel label: normalized historical margin); (g)--(h) canard and
 elevon commands. The safety-filter QP is activated at $t=0$: physical states
 coincide within numerical tolerance, while memory differs. For the Gaussian HA-CBFal comparison, the neutral preparation
 is used for PP-CBF and MA-CBFal traces. The second pulse is magnified in the reserve inset. Regions with $g_\theta<0$ are shaded.}
 \label{fig:stitched}
\end{figure*}

\section{Numerical Evaluation: High-Incidence Envelope Protection}
\label{sec:fep_numerical}

\subsection{ADMIRE Model and Nominal Load-Factor Control}
\label{subsec:admire}

Zero rudder and equal elevons are used for the invariant symmetric ADMIRE
subsystem $\dot x=Ax+Bu$ \cite{forssell2005admire,harkegard2003} at Mach
$0.22$ and $3000$~m. The state $x=[\Delta\alpha\;q]^\top$ and input
$u=[\delta_C\;\delta_{RE}\;\delta_{LE}]^\top$ are trim deviations (rad, rad/s), with
\begin{align}
 A&=\begin{bmatrix}
 -0.5432 & 0.9778\\
 2.6221 & -0.5057
 \end{bmatrix},
 \label{eq:admire_A}\\[0.5mm]
 B&=\begin{bmatrix}
 0.0069 & -0.0866 & -0.0866\\
 1.6532 & -1.2735 & -1.2735
 \end{bmatrix}.
 \label{eq:admire_B}
\end{align}
Using $\Delta n_z=(V/g)(q-\Delta\dot\alpha)$ with $V=72.3$~m/s,
the load-factor output and nominal control law are
$n_z=1+[4.004\;0.164]x+[-0.051\;0.638\;0.638]u$ and
$u_{\rm nom}=-K_xx+K_r(n_{z,{\rm cmd}}-1)$.
The stabilising and trim-allocation gains $K_x,K_r$ are fixed offline;
$u_{\rm nom}$ is clipped componentwise to the actuator box $\U$ before
filtering, with only the safety-filter QP solved online.

\subsection{High-Incidence Reserve Policy and Gaussian-Kernel Synthesis}
\label{subsec:policy}

The safety margin $h_\alpha=\alpha_{\max}-\alpha$ measures the remaining
angle-of-attack allowance to $\alpha_{\max}=26^\circ$. Physical safety
$h_\alpha\geq0$ is supplemented by retention of a fraction $\theta_\alpha$
of its weighted historical value (Corollary~\ref{rem:kernel_policy}). The fraction and
fading law are prescribed for controller-level enforcement in two-state ADMIRE.

The Gaussian kernel $\kappa_G(s)=\kappa_0\exp[-(s/\tau_G)^2]$,
$s\geq0$, is used with $\kappa_0=2$,
$\tau_G=3.0$~s, and $\theta_\alpha=0.40$, giving
$\gamma=\kappa_0(1-\theta_\alpha)/\theta_\alpha=3$.
Assumption~\ref{ass:kernel} is satisfied. An exact finite-dimensional LTI realization
with rational transfer $C_m(sI-A_m)^{-1}B_m+D_m$ is excluded by the kernel's
entire, nonzero Laplace transform tending to zero as $s\to+\infty$.
Memory is evaluated directly from retained history.

For memory origin $t_0$,
\mbox{$q_G(t)=\int_{t_0}^{t}\kappa_G(t-\tau)\dot h_\alpha(\tau)\,d\tau$} and
\mbox{$J_G(t)=\int_{t_0}^{t}\dot\kappa_G(t-\tau)\dot h_\alpha(\tau)\,d\tau$},
where $\dot\kappa_G(s)=-2s\kappa_G(s)/\tau_G^2$. The weighted historical
value and reserve-policy residual are
$\bar h_\alpha=h_\alpha-q_G/\kappa_0$ and
$g_\theta=(1-\theta_\alpha)h_\alpha+\theta_\alpha q_G/\kappa_0$.
With $\dot h_\alpha=-A_\alpha x-B_\alpha u$, the implemented constraint is
\begin{align}
Q_{G,\alpha}\triangleq{}&-(\kappa_0+\gamma)B_\alpha u
-(\kappa_0+\gamma)A_\alpha x\notag\\
&+J_G+\beta q_G+\beta\gamma h_\alpha\geq0.
\label{eq:gaussian_row_admire}
\end{align}
With $Q_{G,\alpha}=(\kappa_0+\gamma)(\dot g_\theta+\beta g_\theta)$ and
$\Lambda_G=J_G+\beta q_G$, $\tfrac12\|u-u_{\rm nom}\|_2^2$ is minimized subject to
\eqref{eq:gaussian_row_admire} and the actuator box.

For the policy-preserving CBF (PP-CBF), $\dot h_\alpha+k h_\alpha\geq0$
is enforced by the same filter. With case tracking RMSE $e_s(k)$, the
gain is calibrated on $\mathcal S$ (Section~\ref{subsec:stress}) by
\begin{equation}
\begin{gathered}
k_{\rm pp}^\star\in\argmin_{k\in[0.01,4]}
\frac{1}{|\mathcal S|}\sum_{s\in\mathcal S}e_s(k)\\
\mathrm{s.\ t.}\quad h_{\alpha,s}(t;k)\geq0,\quad
g_{\theta,s}(t;k)\geq0,
\end{gathered}
\label{eq:pp_gain_optimization}
\end{equation}
By scalar search and feasible-side rounding,
$k_{\rm pp}=0.2393$~s$^{-1}$ is obtained.

For the moving-average control barrier functional (MA-CBFal),
the structure in \cite[Ex.~1]{kiss2023cbf} is instantiated with current and
distributed values of the same safety margin:
\[
\Phi(t)=(1-\eta)h_\alpha(t)+\frac{\eta}{T}
\int_{t-T}^{t}h_\alpha(s)\,ds,
\qquad 0<\eta<1.
\]
The inequality $\dot\Phi+\beta_{\rm MA}\Phi\geq0$ is enforced by the same QP,
using $\dot\Phi=(1-\eta)\dot h_\alpha+\eta[h_\alpha(t)-h_\alpha(t-T)]/T$.
Under the same calibration objective and constraints, $\eta=0.55$, $T=3$~s,
and $\beta_{\rm MA}=0.2777$~s$^{-1}$ are obtained by finite search.

The HA-CBFal gain $\beta=3.8341$~s$^{-1}$ is retained from a single-exponential
reference design and fixed before comparator calibration.
All gains are then frozen, the nominal controller is shared, and reserve-policy
residuals are assessed with $5\times10^{-5}$-rad tolerance.

\subsection{History-Prepared Twin-Square Evaluation}
\label{subsec:doublet}

Three physically safe, unfiltered preparations are initialized at
$t_0=-9$~s with $q_G(t_0)=0$. Activation endpoints are matched within
$8.6\times10^{-4}$ degrees in incidence and $5.8\times10^{-5}$~rad/s in
pitch rate; each endpoint and full memory are retained. At activation,
$q_G=\{-0.19919,-0.11270,-0.05067\}$ and
$g_\theta(0)=\{0.00131,0.01861,0.03101\}$~rad for the adverse, neutral,
and favorable histories. The reserve policy and QP are activated at $t=0$.
The activation reconstruction is split to retain a nonnegative preparation
contribution, with auxiliary initialization from $g_\theta(0)\geq0$.
Two $2.6$-g, $2$-s squares with $1$-s smooth transitions are commanded beyond
the protected load-factor capability, separated by a $1$-s release for margin recovery.

Both requirements are preserved to tolerance by PP-CBF and HA-CBFal,
with lower HA-CBFal tracking RMSE (Fig.~\ref{fig:stitched}, Table~\ref{tab:doublet}).
Under MA-CBFal, physical safety and $\Phi\geq0$ are preserved, but the reserve requirement
is violated for all three preparations (neutral minimum: $-0.0149$~rad); it is
therefore not certified by this functional. The 1-s release was selected
from a recovery-interval sweep outside the original calibration set;
no controller is retuned.

\begin{table}[!t]
\caption{One-second-recovery results (neutral preparation for comparator traces).}
\label{tab:doublet}
\centering
\setlength{\tabcolsep}{4pt}
\begin{tabular}{@{}lrrr@{}}
\toprule
Controller & $\alpha_{\rm pk}$ [deg] & $\min g_\theta$ [rad]
& $e_{n_z,{\rm rms}}$ [g]\\
\midrule
Nominal (unsafe ref.) & 30.25 & $-0.1076$ & 0.3406\\
PP-CBF & 24.27 & $+0.0042$ & 0.4871\\
MA-CBFal & 24.29 & $-0.0149$ & 0.3354\\
\midrule
HA-CBFal, adverse & 24.49 & $-1.37\!\times\!10^{-8}$ & 0.3292\\
HA-CBFal, neutral & 24.61 & $1.20\!\times\!10^{-7}$ & 0.3243\\
HA-CBFal, favorable & 24.73 & $2.99\!\times\!10^{-7}$ & 0.3213\\
\bottomrule
\end{tabular}
\end{table}

\begin{figure}[!t]
 \centering
 \includegraphics[width=\columnwidth]{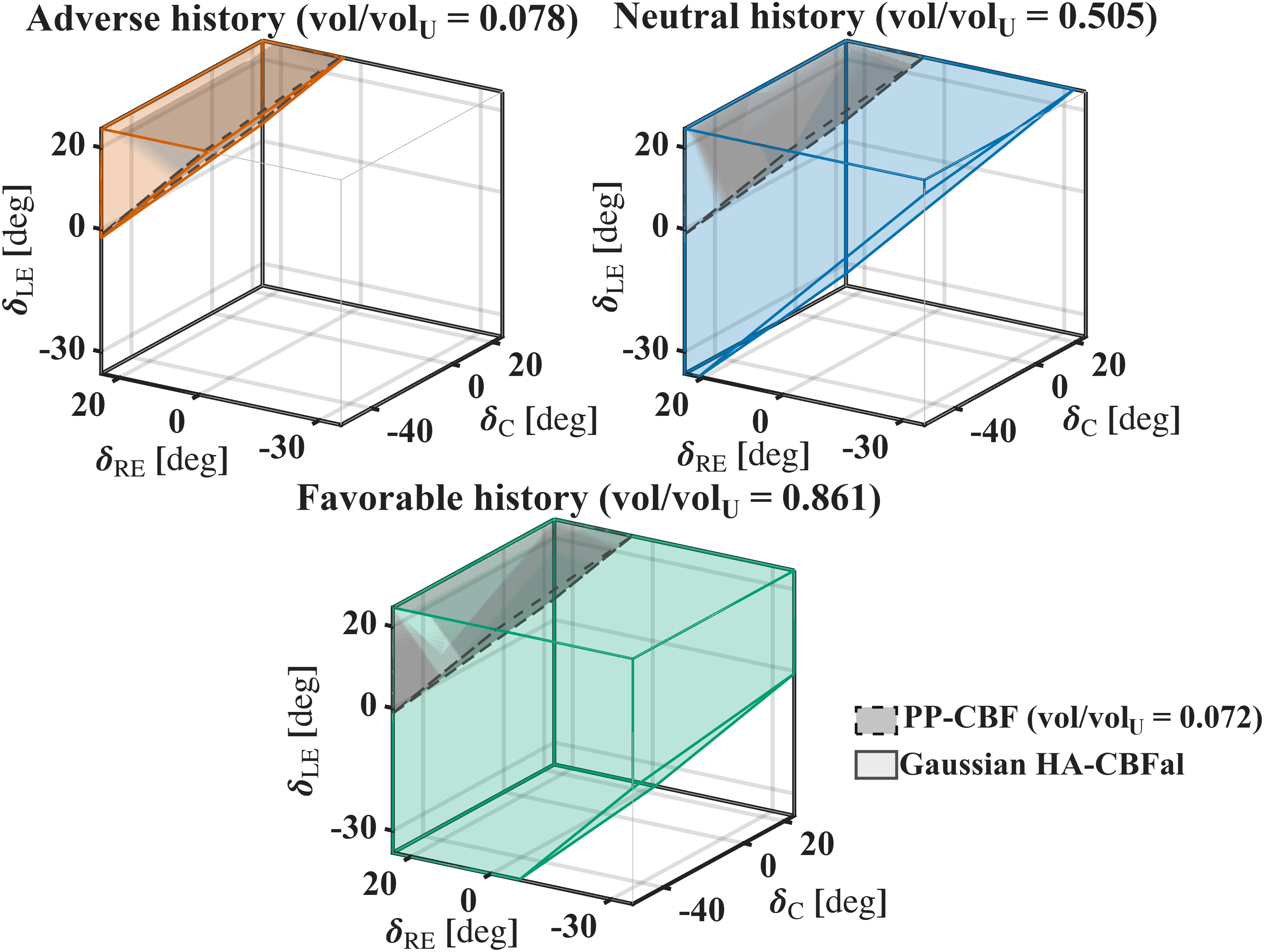}
 \caption{Three-dimensional instantaneous incidence-barrier input polytopes
 at the common symmetric checkpoint $x^\star$, intersected with the actuator
 box. Dashed: PP-CBF; solid: HA-CBFal. State, current command,
 and box are fixed; only the retained history differs. Equal elevon commands are used in the simulated trajectories.}
 \label{fig:polytopes}
\end{figure}

At the mean checkpoint $x^\star=(9.3709^\circ,0.15047~\mathrm{rad/s})$
($\alpha^\star=22.0709^\circ$ absolute), the independently tuned PP-CBF
set is compared with HA-CBFal sets ordered by
$\Lambda_G=\{-0.70008,-0.42244,-0.20990\}$ (Fig.~\ref{fig:polytopes}).
Normalized volumes $\{0.07763,0.50500,0.86112\}$ versus $0.07197$ for PP-CBF
are obtained by half-space intersection. By
Proposition~\ref{prop:history_parametrized_sets}, the sets are nested,
with the adverse set as their intersection.

\subsection{Calibration and Maneuver Tests}
\label{subsec:stress}

The original calibration set $\mathcal S$ is formed from twin-square peaks
$\{2.4,2.6,2.8,3.0\}$~g, recoveries $\{2,3,4\}$~s, and three histories
(36 cases); 1-s recovery is added for 48 cases. On each set, PP-CBF and
MA-CBFal mean RMSE is minimized subject to $h_\alpha,g_\theta\geq0$ and QP
feasibility throughout every trajectory; the HA-CBFal gain is retained.
After 48-case calibration, all parameters are frozen for 27 exploratory
tests outside both sets: nine additional twin-square, triple-square, and
sinusoidal profiles under three histories.

With the original calibration, both constraints are met by all filters
on $\mathcal S$; HA-CBFal RMSE is lower in every case (mean per-case reductions:
$27.2\%$ versus PP-CBF, $7.0\%$ versus MA-CBFal). With 48-case calibration,
both constraints are met by MA-CBFal throughout calibration, but the
reserve requirement is violated in 17 of the 27 tests despite $h_\alpha,\Phi\geq0$.
Both constraints are met by HA-CBFal to tolerance in all 27 tests.

Policy reuse is tested with Gaussian memory on $h_\alpha$ and
$h_\alpha+h_\alpha^3/a^2$ ($a=0.25$~rad), and uniform historical
weighting of $h_\alpha$ induced by $\kappa_U(s)=2\max\{1-s/T_U,0\}$,
$T_U=3\sqrt\pi$~s. A common $\theta_c\simeq0.2567$ is selected from
preparation data for admissible initialization, replacing $0.4$;
$\gamma$ and policy terms are recomputed analytically, with $\beta$
retained. Both constraints are met to tolerance in all 27 maneuvers per policy.

Reserve satisfaction is checked during MA-CBFal calibration, whereas only
its functional inequality is enforced online; further tuning or redesign
is therefore needed for the new violations. No HA-CBFal gain retuning is
required in these tests; joint preservation is certified under the
hypotheses of Theorem~\ref{thm:physical_transfer}.\footnote{Model and actuator
data, comparator calibration, reference-gain selection, and additional results:
\href{https://doi.org/10.5281/zenodo.23284445}{doi:10.5281/zenodo.23284445}.}

\section{Conclusion}

Prescribed reserve requirements and physical safety are jointly certified by an explicit
construction for integer-order plants under the stated conditions. The same
implication proof is applicable across admissible kernels without
finite-dimensional memory realization. Uniform history-blind enforcement is characterized by
the exact worst-history intersection. In ADMIRE, tracking is improved on
the calibration set, and both constraints are preserved in additional
maneuvers with fixed HA-CBFal gains. Calibration success is distinguished from
certified reserve enforcement by the observed MA-CBFal violations despite
functional nonnegativity. Nonlinear flight models and bounded-input
feasibility maintenance are left for future work.

\bibliographystyle{IEEEtran}
\bibliography{references}

@article{ames2017cbfqp,
  author  = {Ames, Aaron D. and Xu, Xiangru and Grizzle, Jessy W. and Tabuada, Paulo},
  title   = {Control Barrier Function Based Quadratic Programs for Safety Critical Systems},
  journal = {IEEE Trans. Autom. Control},
  volume  = {62},
  number  = {8},
  pages   = {3861--3876},
  year    = {2017}
}

@inproceedings{ames2019cbf,
  author    = {Ames, Aaron D. and Coogan, Samuel and Egerstedt, Magnus and Notomista, Gennaro and Sreenath, Koushil and Tabuada, Paulo},
  title     = {Control Barrier Functions: Theory and Applications},
  booktitle = {Proc. 18th Eur. Control Conf. (ECC)},
  pages     = {3420--3431},
  year      = {2019}
}

@inproceedings{lackey1992pitchmargin,
  author    = {J. Lackey and C. Hadfield},
  title     = {Pitch Control Margin at High Angle of Attack---Quantitative Requirements (Flight Test Correlation With Simulation Predictions)},
  booktitle = {Proc. 6th AIAA Biennial Flight Test Conf.},
  year      = {1992},
  note      = {{AIAA} Paper 92-4107}
}

@inproceedings{brandon1991dynamicstall,
  author    = {J. M. Brandon},
  title     = {Dynamic Stall Effects and Applications to High Performance Aircraft},
  booktitle = {AGARD Special Course on Aircraft Dynamics at High Angles of Attack: Experiments and Modelling},
  year      = {1991},
  note      = {{AGARD}-R-776}
}

@inproceedings{prajna2005delay,
  author    = {Prajna, Stephen and Jadbabaie, Ali},
  title     = {Methods for Safety Verification of Time-Delay Systems},
  booktitle = {Proc. 44th IEEE Conf. Decision Control Eur. Control Conf. (CDC-ECC)},
  pages     = {4348--4353},
  year      = {2005}
}

@inproceedings{orosz2019safety,
  author    = {Orosz, Gabor and Ames, Aaron D.},
  title     = {Safety Functionals for Time Delay Systems},
  booktitle = {Proc. Amer. Control Conf. (ACC)},
  pages     = {4374--4379},
  year      = {2019}
}

@inproceedings{ren2021razumikhin,
  author    = {Ren, Wei},
  title     = {Razumikhin-Type Control Lyapunov and Barrier Functions for Time-Delay Systems},
  booktitle = {Proc. 60th IEEE Conf. Decision Control (CDC)},
  pages     = {5471--5476},
  year      = {2021}
}

@article{ren2022safe,
  author  = {Ren, Wei and Jungers, Rapha{\"e}l M. and Dimarogonas, Dimos V.},
  title   = {Razumikhin and Krasovskii Approaches for Safe Stabilization},
  journal = {Automatica},
  volume  = {146},
  pages   = {110563},
  year    = {2022}
}

@article{kiss2023cbf,
  author  = {Kiss, A. K. and Molnar, T. G. and Ames, Aaron D. and Orosz, Gabor},
  title   = {Control Barrier Functionals: Safety-Critical Control for Time Delay Systems},
  journal = {Int. J. Robust Nonlinear Control},
  volume  = {33},
  number  = {12},
  pages   = {7282--7309},
  year    = {2023}
}

@article{zhu2023first,
  author  = {Zhu, Z. and Huang, P. and Zhang, X. and Chai, Y. and Song, Z.},
  title   = {First Attempt of Barrier Functions for {Caputo}'s Fractional-Order Nonlinear Dynamical Systems},
  journal = {Sci. China Inf. Sci.},
  volume  = {66},
  number  = {7},
  pages   = {179205},
  year    = {2023}
}

@article{zhu2025caputo,
  author  = {Zhu, Z. and Shen, B. and Yao, L. and Chai, Y. and Song, Z.},
  title   = {{Caputo} Barrier Functions and Their Applications to the Safety, Safety-and-Stability, and Input-to-State Safety of a Class of Fractional-Order Systems},
  journal = {Mathematics},
  volume  = {13},
  number  = {8},
  pages   = {1215},
  year    = {2025}
}

@article{sabatier2020continuous,
  author  = {Sabatier, Jocelyn},
  title   = {Fractional-Order Derivatives Defined by Continuous Kernels: Are They Really Too Restrictive?},
  journal = {Fractal Fract.},
  volume  = {4},
  number  = {3},
  pages   = {40},
  year    = {2020}
}

@book{GripenbergLondenStaffans1990,
  author    = {Gripenberg, Gustaf and Londen, Stig-Olof and Staffans, Olof},
  title     = {Volterra Integral and Functional Equations},
  publisher = {Cambridge Univ. Press},
  address   = {Cambridge, U.K.},
  year      = {1990}
}

@techreport{forssell2005admire,
  author      = {Forssell, Lars and Nilsson, Ulf},
  title       = {{ADMIRE}: The Aero-Data Model in a Research Environment, Version 4.0---Model Description},
  institution = {Swedish Defence Research Agency (FOI)},
  number      = {FOI-R--1624--SE},
  year        = {2005}
}

@phdthesis{harkegard2003,
  author  = {H{\"a}rkeg{\aa}rd, Ola},
  title   = {Backstepping and Control Allocation with Applications to Flight Control},
  school  = {Link{\"o}ping Univ.},
  address = {Link{\"o}ping, Sweden},
  number  = {820},
  year    = {2003}
}

@manual{abb2015umc,
  organization = {{ABB}},
  title = {{UMC100.3} Universal Motor Controller: Technical Description},
  year = {2015},
  month = jul,
  note = {Doc. 2CDC135032D0203, Sec. 4}
}

@article{jang2023viscoelastic,
  author = {Jang, Yongseok and Shaw, Simon},
  title = {A Priori Analysis of a Symmetric Interior Penalty Discontinuous {Galerkin} Finite Element Method for a Dynamic Linear Viscoelasticity Model},
  journal = {Comput. Methods Appl. Math.},
  volume = {23},
  number = {3},
  pages = {647--669},
  year = {2023},
  doi = {10.1515/cmam-2022-0201}
}

\end{document}